\documentclass[11pt]{article}
\usepackage[letterpaper,margin=1in]{geometry}
\usepackage{amsmath, amssymb, amsthm, amsfonts}
\usepackage{cite}

\usepackage{bbm}

\usepackage{ifthen}
\usepackage{tikz}
\usetikzlibrary{positioning,decorations.pathreplacing}

\usepackage{appendix}
\usepackage{graphicx}
\usepackage{color}
\usepackage{color}
\usepackage{algorithm}
\usepackage{algorithm}
\usepackage[noend]{algpseudocode}
\usepackage{epstopdf}
\usepackage{wrapfig}
\usepackage{paralist}
\usepackage{wasysym}
\usepackage[textsize=tiny, disable]{todonotes}

\usepackage{framed}
\usepackage[framemethod=tikz]{mdframed}
\usepackage[bottom]{footmisc}
\usepackage{enumitem}
\setitemize{noitemsep,topsep=3pt,parsep=3pt,partopsep=3pt}
\usepackage[font=small]{caption}
\usepackage{xspace}

\newtheorem{theorem}{Theorem}[section]
\newtheorem{lemma}[theorem]{Lemma}
\newtheorem{meta-theorem}[theorem]{Meta-Theorem}

\newtheorem{corollary}[theorem]{Corollary}

\newtheorem{definition}[theorem]{Definition}

\definecolor{darkgreen}{rgb}{0,0.5,0}
\usepackage{hyperref}
\hypersetup{
	unicode=false,          
	colorlinks=true,        
	linkcolor=red,          
	citecolor=darkgreen,        
	filecolor=magenta,      
	urlcolor=cyan           
}
\usepackage[capitalize, nameinlink]{cleveref}
\crefname{theorem}{Theorem}{Theorems}
\Crefname{lemma}{Lemma}{Lemmas}
\Crefname{remark}{Remark}{Remarks}
\Crefname{observation}{Observation}{Observations}
\Crefname{claim}{Claim}{Claims}
\algnewcommand\algorithmicswitch{\textbf{switch}}
\algnewcommand\algorithmiccase{\textbf{case}}

\algdef{SE}[SWITCH]{Switch}{EndSwitch}[1]{\algorithmicswitch\ #1\ \algorithmicdo}{\algorithmicend\ \algorithmicswitch}%
\algdef{SE}[CASE]{Case}{EndCase}[1]{\algorithmiccase\ #1}{\algorithmicend\ \algorithmiccase}%
\algtext*{EndSwitch}%
\algtext*{EndCase}%

\newcommand{\congest}{$\mathsf{CONGEST}$\xspace}
\newcommand{\local}{$\mathsf{LOCAL}$\xspace}

\newcommand{\poly}{\operatorname{\text{{\rm poly}}}}

\DeclareMathOperator{\E}{\mathbb{E}}

\newcommand{\Abs}[1]{\left\lvert#1\right\rvert}

\renewcommand{\paragraph}[1]{\medskip\noindent {\bf #1.}}

\begin{document}
\date{}
\title{Distributed MIS with Low Energy and Time Complexities}

\author{
  Mohsen Ghaffari \\
  \small MIT \\
  ghaffari@mit
  \and
  Julian Portmann\\
  \small ETH Zurich \\
  pjulian@inf.ethz.ch
}

\maketitle

\begin{abstract}
	We present randomized distributed algorithms for the maximal independent set problem (MIS) that, while keeping the time complexity nearly matching the best known, reduce the energy complexity substantially.
	These algorithms work in the standard CONGEST model of distributed message passing with $O(\log n)$ bit messages.
	The time complexity measures the number of rounds in the algorithm.
	The energy complexity measures the number of rounds each node is awake; during other rounds, the node sleeps and cannot perform any computation or communications.

	Our first algorithm has an energy complexity of $O(\log\log n)$ and a time complexity of $O(\log^2 n)$.
	Our second algorithm is faster but slightly less energy-efficient: it achieves an energy complexity of $O(\log^2 \log n)$ and a time complexity of $O(\log n \cdot \log\log n \cdot \log^* n)$.
	Thus, this algorithm nearly matches the $O(\log n)$ time complexity of the state-of-the-art MIS algorithms while significantly reducing their energy complexity from $O(\log n)$ to $O(\log^2 \log n)$.
\end{abstract}

\section{Introduction and Related Work}
We present randomized distributed algorithms for the maximal independent set (MIS) problem that perform well in both time and energy complexity measures.
The MIS problem has been one of the central problems in the study of distributed graph algorithms.
Much of the research on this problem has focused on the time complexity measure, i.e., the maximum time each node spends running the distributed algorithm.
The state-of-the-art time complexity is $O(\log n)$~\cite{luby86, alon1986fast}; this has remained the best-known bound for general graphs for nearly four decades.

Another important measure, which is vital in applications such as sensor and wireless networks, is the amount of energy spent by each node.
Roughly speaking, the energy consumed by a node can be approximated as being proportional to the time that the node is active in the distributed algorithm and performs computations or communications.
In contrast to the time complexity measure, which has been widely studied, we know considerably less about the energy complexity measure.
Standard algorithms have their energy complexity equal to the time complexity because usually it is unclear how to let each node be inactive during a substantial part of the algorithm while retaining the algorithm's correctness and time complexity.
In this paper, we show MIS algorithms whose time complexity nearly matches the best-known bound while achieving a substantially reduced energy complexity.
We next state our formal results, after reviewing the model and the precise complexity measures, and the state of the art.

\subsection{Setup}
\paragraph{Distributed Model}
We work with the standard synchronous model of distributed message passing, often referred to as the \congest model~\cite{peleg00}.
The network is abstracted as an $n$-node undirected graph $G=(V, E)$, where each node represents one computer.
These nodes operate in synchronous rounds.
Per round, each node performs some computations on the data that it holds, and then it can send one message to each of its neighbors in $G$.
The standard assumption is that each message has size $B=O(\log n)$, which is sufficient to describe constant many nodes or edges and values polynomially bounded in $n$.
The variant where messages are allowed to be unbounded is often called the \local model~\cite{linial92}.
The messages are delivered by the end of the round.
Then, the algorithm proceeds to the next round.
In formulating distributed graph problems, the assumption is that the nodes do not know the global topology of $G$.
Initially, each node knows only its own neighbors, perhaps along with some bounds on global parameters, e.g., a polynomial bound on $n$.
At the end of the algorithm, each node should know its part of the output, e.g., in the problem of computing an MIS, each node should know whether it is in the computed set or not.

\paragraph{Time and Energy Complexities}
The primary measure of interest has been the \textit{time complexity}, sometimes also called round complexity.
This is the maximum number of rounds each node runs the algorithm, i.e., the time from the start of the algorithm until the node terminates.
The \textit{energy complexity} is a sharper version that measures only the rounds in which a node is actually active and does some computations or communications.
Let us say per round each node is either \textit{active/awake} or \textit{sleeping}.
During a round in which the node sleeps, it cannot perform any computation, and it does not send or receive any messages.
Further, it also cannot be woken up by any other node, it only wakes up after a scheduled amount of time.
The \textit{energy complexity} of the algorithm is defined as the maximum number of rounds each node is active/awake.
Clearly, the energy complexity is upper bounded by the time complexity.
We believe it is interesting---from both the perspective of applications and the view of technical algorithm-design challenges---to devise algorithms that have a much better energy complexity, ideally without a significant sacrifice in the time complexity.

The energy complexity measure---considering awake or sleeping rounds and measuring the number of awake rounds---has been studied for over two decades.
See, e.g., \cite{jurdzinski2002efficient, jurdzinski2002energy, kardas2013energy, chang2017exponential, chang2018energy, chang2020energy, chatterjee2020sleeping, barenboim2021deterministic, augustine2022brief, ghaffari2022average, dufoulon2022sleeping}.
Much of this was in the context of the radio network model\footnote{In the radio network model, per round, each node is transmitting or listening, and a listening node receives a message only if there is exactly one transmitting neighbor.
  When adding the energy complexity consideration, we again have awake or sleeping rounds, and a node can be transmitting or listening only in awake rounds.}.
More recent work investigates this measure in the context of the more basic synchronous message-passing models of distributed graph algorithms, namely the \congest and \local models~\cite{chatterjee2020sleeping, barenboim2021deterministic, augustine2022brief, ghaffari2022average, dufoulon2022sleeping}.

We remark that sometimes the word \textit{awake-complexity} is used instead of energy-complexity.
Moreover, sometimes this variant of the \congest model that considers sleeping nodes is called \textit{the sleeping model}~\cite{chatterjee2020sleeping}.
In this article, we decided to refrain from the latter terminology as we think sleeping is a separate aspect and it can be considered for any of the distributed models, as it has been done for the radio network model, the \congest model, the \local model, etc.

\subsection{Prior Work}
Let us start with time complexity alone.
Luby's $O(\log n)$ round randomized MIS algorithm~\cite{luby86, alon1986fast} remains state of the art in time complexity.
Much research has been done on time complexity, and there are faster algorithms for special graph families; see e.g.,~\cite{barenboim2016locality, gmis, ghaffari2019distributed}.
The best known lower bound is $\Omega(\sqrt{\log n/\log\log n })$~\cite{kuhn16_jacm}.
For the restriction to deterministic algorithms, a best lower bound of $\Omega(\log n/\log\log n)$ is known~\cite{balliu2021lower}, and there are deterministic algorithms within a polynomial of this, i.e., with round complexity $\poly (\log n)$~\cite{rozhovn2019nd, ghaffari2021improvedNC,faour2022localRounding}.
All of these have a high energy complexity, and in particular, Luby's $O(\log n)$-round algorithm also has energy complexity $O(\log n)$.
Ideally, we would like the energy complexity to be much lower.

We next discuss prior results that studied MIS from both energy and time views.
Some previously known algorithms can provide bounds on the \textit{average} energy spent per node, i.e., the average over all nodes of the summation of their awake rounds.
This is instead of the energy complexity defined above that measures the maximum energy spent by each node, i.e., the maximum over all nodes of the number of rounds in which this node is awake.
Chatterjee, Gmyr, and Pandurangan\cite{chatterjee2020sleeping} presented an MIS algorithm that achieves an $O(1)$ average energy and $O(\log^{3.41} n)$ time complexity.
Ghaffari and Portmann~\cite{ghaffari2022average} improved this to an algorithm that achieves an $O(1)$ average energy with a time complexity of $O(\log n)$, which matches the state-of-the-art time complexity of MIS.
However, these algorithms still have $\Omega(\log n)$ energy complexity; that is, some nodes have to be awake and active during at least $\Theta(\log n)$ rounds.

In a recent work, Dufoulon, Moses Jr., and Pandurangan~\cite{dufoulon2022sleepingarxiv} presented algorithms that achieve a sublogarithmic energy complexity.
Their first algorithm reaches an energy complexity of $O(\log\log n)$, but at the expense of having a time complexity of $\poly(n)$.
Note that this time complexity is exponentially higher than the state-of-the-art $O(\log n)$ complexity of MIS.
Their second algorithm achieves an energy complexity of $O(\log \log n \cdot \log^*n)$ and a time complexity of $\tilde{O}(\log^3 n)$.
However, it should be noted that this algorithm needs messages of size at least $\Omega(\log^3 n)$, as it performs some non-trivial topology gathering.
If one limits the algorithm to the more standard $O(\log n)$ message size of the \congest model, then the algorithm's energy-complexity blows up to $\tilde{\Omega}(\log^2 n)$, because some nodes need to send $O(\log^3 n)$ bits of information.
This energy complexity would make the result uninteresting because Luby's classic algorithm~\cite{luby86} has better time and energy complexities.
Concurrent to our work, Dufoulon, Moses Jr., and Pandurangan~\cite{dufoulon2022sleeping} have also improved the message size of their algorithms to $O(\log n)$, while maintaining the same round and energy complexities.
This appears in the conference version~\cite{dufoulon2022sleeping}, as well as in an updated version of their preprint~\cite{dufoulon2022sleepingarxiv}.
We also comment that Hourani et al.~\cite{hourani2022awake} gave some related results for MIS in the very special case of random graphs embedded in low-dimensional geometric spaces, but their results are subsumed by those of \cite{dufoulon2022sleeping}.

\subsection{Our Results}
We are interested in algorithms that are (almost) as fast as the state-of-the-art but use substantially less energy.
That is, algorithms that achieve good energy complexity without (considerably) sacrificing the more standard measure of time complexity.
We remark that this is technically a far more challenging problem than if we give the algorithm much more time.
The algorithms we desire have to achieve time complexity almost equal to the best known, while each node sleeps almost all the time, except for an exponentially tiny fraction of the rounds.

Our first algorithm obtains the best known $O(\log\log n)$ bound of the energy complexity while running in $O(\log^2 n)$ rounds.
\begin{theorem}
	\label[theorem]{thm:alg1}
	There is a randomized distributed MIS algorithm with time complexity $O(\log^2 n)$ and energy complexity $O(\log\log n)$.
	That is, the algorithm runs for $O(\log^2 n)$ rounds, each node is awake in only $O(\log\log n)$ rounds, and the algorithm computes an independent set of vertices such that, with high probability, this set is maximal.
\end{theorem}

Our second algorithm has time complexity almost matching the long-standing $O(\log n)$ bound~\cite{luby86}, while achieving an energy complexity of $O(\log^2\log n)$.
\begin{theorem}
	\label[theorem]{thm:alg2}
	There is a randomized distributed MIS algorithm with time complexity $O(\log n \cdot \log\log n \cdot \log^* n)$ and energy complexity $O(\log^2 \log n)$.
	That is, the algorithm runs for $O(\log n \cdot \log\log n \cdot \log^* n)$ rounds, each node is awake in only $O(\log^2\log n)$ rounds, and the algorithm computes an independent set of vertices such that, with high probability, this set is maximal.
\end{theorem}

As an additional point, we also show that---with the addition of one suitable algorithmic module---these algorithms can ensure that the average energy spent over all nodes remains $O(1)$, while keeping the time and energy complexities unchanged.
Hence, these augmented algorithms simultaneously match the $O(1)$ average-energy bounds of the algorithms of \cite{ghaffari2022average, chatterjee2020sleeping}; recall that the latter do not provide sublogarithmic (worst-case) energy complexity.
This extension to constant node-averaged energy appears in \Cref{sec:average}.



\section{Algorithm 1}
\label{sec:alg1}
In this section, we present our first algorithm, which has time complexity $O(\log^2 n)$ and energy complexity $O(\log \log n)$.
The algorithm is split into three phases:
The first phase will reduce the maximum degree to $\poly (\log n)$, the second phase will break the graph into small components, also known as \emph{shattering}, and the third phase will be a clean-up that deals with the remaining small components.
Finally, we will put everything together and prove \Cref{thm:alg1}.

\subsection{Phase I}
\label{sec:alg1p1}
In this first phase, our goal is to compute an independent set of nodes such that, after removing these nodes and their neighbors from the graph, the remaining graph has maximum degree at most $O(\log^2 n)$.
This is formalized in the following lemma.
\begin{lemma}
	\label[lemma]{lem:alg1p1}
	There is an algorithm that, given a graph of maximum degree $\Delta$, computes an independent set $S$ of nodes with the following properties: removing all nodes of $S$ together with their neighbors from the graph results in a graph with maximum degree $O(\log^2 n)$.
	Furthermore, the algorithm runs in $O(\log \Delta \cdot \log n)$ rounds of the \congest model, and each node is awake for at most $O(\log \log n)$ rounds.
	The algorithm succeeds with probability at least $1 - n^{-10}$.
\end{lemma}
We will first describe the algorithm as if all nodes were always active.
Afterward, we discuss how our modifications allow the algorithm to be executed with low energy complexity.

\paragraph{Regularized Luby}
To achieve this goal, we will use a slowed-down variant of Luby's algorithm \cite{luby86}, sometimes also called regularized Luby.
In its basic form, the algorithm works as follows:
In the first iteration, each node is marked with probability $\frac{1}{10 \Delta}$ per round.
Then, each round consists of two sub-rounds, which each correspond to one round of communication in the network:
In the first sub-round, each marked node $v$ informs all its neighbors that it is marked, and learns if any of its neighbors are marked.
In the second sub-round, if $v$ is marked and has no marked neighbors, $v$ joins the MIS, removes itself from the graph, and informs all its neighbors that they will be removed from the graph.
This is repeated for $c \cdot \log n$ rounds.
We can argue that afterward, no node of degree more than $\frac{\Delta}{2}$ remains with high probability.
More generally, in iteration $i$, nodes are marked with probability $\frac{2^i}{10 \Delta}$ for $c \cdot \log n$ rounds.
We can show that at the end of iteration $i$, the maximum degree drops to $\frac{\Delta}{2^i}$ with high probability.
Finally, after $O(\log \Delta)$ iterations, all nodes that are isolated are added to the MIS.

\paragraph{Modifications} We will modify the algorithm as follows: The first modification is that we let nodes be marked at most once.
That is, after the first time a node is marked, if it did not join the MIS, this node is considered \emph{spoiled} and will remain spoiled until the end of this first phase of the entire algorithm.
The second modification is that we will not perform $\log \Delta$ iterations but only $\log \Delta - 2 \log \log n$ many.
Both these modifications also affect which guarantees we expect from the algorithm:
In the end, we will show that the graph of all nodes that were not removed due to themselves or their neighbor joining the MIS has degree at most $O(\log^{2} n)$.

\paragraph{Analysis}
First, recall that we call a node that was marked but did not subsequently join the MIS \emph{spoiled} for the remaining part of the algorithm.
We call a node \emph{active} if neither itself nor any of its neighbors were added to the MIS (yet), and \emph{inactive} otherwise.
Note that nodes that are spoiled can also become inactive and active nodes can be spoiled.
We will show that the following invariants are maintained w.h.p.:
At the end of iteration $i$, we have that
\begin{itemize}
	\item $A(i)$: for any active node, at most $(i+1) \cdot C \log n$ neighbors of it are active and spoiled, and
	\item $B(i)$: each active node has at most $\frac{\Delta}{2^{(i+1)}}$ neighbors that are active and not spoiled.
\end{itemize}
We will prove these two invariants separately and both by induction.
At the beginning of the algorithm, these invariants are trivially satisfied.
For simplicity, we consider these to be the statements $A(-1)$ and $B(-1)$.
Thus, we focus our attention on the case of $i \geq 0$.
\begin{lemma}
	\label[lemma]{lem:Bi}
	For any $i \geq 0$, we have that $B(i-1)$ implies $B(i)$ with probability $1 - n^{-11}$.
\end{lemma}
\begin{proof}
	Let us focus on one round in iteration $i$, and one active node $v$ which has more than $\frac{\Delta}{2^{(i+1)}}$ active and non-spoiled neighbors.
	We will show that either (a) this number of neighbors drops below the threshold, or (b) node $v$ becomes inactive with high probability.

	First, note that the number of active and non-spoiled neighbors can only decrease, by nodes becoming spoiled or inactive due to being removed by their neighbors.
	By $B(i-1)$ we know that $v$, and any other active node, has at most $\frac{\Delta}{2^i}$ active and non-spoiled neighbors throughout iteration $i$.
	Thus, for a given neighbor $u$, the probability that $u$ is the only marked neighbor in both its own and $v$'s neighborhood is at least
	\[
		\frac{2^i}{\Delta} \cdot \left( 1 - \frac{2^i}{\Delta} \right)^{\frac{2 \Delta}{2^i}} \geq \frac{2^i}{\Delta} \cdot 4^{- 2/10} \geq \frac{2^i}{20 \Delta}.
	\]
	Since these events are independent for all neighbors of $v$, the probability that at least one of them occurs is at least $1/40$.
	This is also a lower bound on the probability that $v$ is removed: the reason is that if any of its neighbors is the only marked node in its neighborhood, it will join the MIS and remove $v$ from the graph.
	Thus, $v$ is removed with a probability of at least $1/40$.

	Assuming that the number of active and non-spoiled neighbors never drops below $\frac{\Delta}{2^i}$, this is true for each round.
	Since rounds within one iteration are independent, the probability that $v$ is still active and has more than $\frac{\Delta}{2^{(i+1)}}$ neighbors after $c \log n$ rounds is at most $1 - n^{-12}$ for $c$ sufficiently large.
	By a union bound over all nodes, we have that $B(i)$ thus holds with probability $1 - n^{-11}$.
\end{proof}
\begin{lemma}
	For any $i \geq 0$, we have that $A(i-1)$ and $B(i-1)$ imply $A(i)$ with probability $1 - n^{-11}$.
	\label[lemma]{lem:Ai}
\end{lemma}
\begin{proof}
	We will focus on one node $v$ and union bound over all nodes in the last step.
	Since we know that $A(i-1)$ is true, we just need to show that in iteration $i$, at most $C \log n$ active and non-spoiled neighbors of $v$ become spoiled.
	In fact, we will show that at most as many neighbors are marked, which are the only nodes that can become spoiled.
	From $B(i-1)$, we get that $v$ has at most $\frac{\Delta}{2^i}$ active and non-spoiled neighbors at the beginning of iteration $i$.
	For each of these neighbors $u$, let $X_u$ be the event that $u$ is ever marked in iteration $i$.
	We have that $\Pr[X_u = 1] \leq c \log n \frac{2^i}{\Delta}$.
	Let $X = \sum_{u \in N(v)} X_u$ be the number of neighbors marked in iteration $i$.
	We have that $\E[X] \leq c \log n$.
	Since all $X_u$ are independent, we can use Chernoff's Bound to obtain:
	\begin{align*}
		\Pr[X \geq C \log n]
		 & \leq \Pr \left[ X - \E[X] \geq \frac{(C - c) \log n}{\E[X]} \E[X] \right]                                 \\
		 & \leq \exp \left( - \frac{\frac{(C - c)^2 \log^2 n}{\E[X]^2} \E[X]}{2 + \frac{(C-c) \log n}{\E[X]}}\right)
		  \leq \exp \left( - \frac{(C-c)^2 \log^2 n}{(C + c) \log n} \right) \leq n^{-12}
	\end{align*}
	Where we assumed that $C \geq 100c$ in the last inequality.
	By a union bound over all nodes, $A(i)$ holds with probability $1 - n^{-11}$.
\end{proof}
Having proven both those induction steps, we can put them together to obtain the following:
\begin{lemma}
	\label[lemma]{lem:p1degree}
	After $T = \log \Delta - \log \log n$ iterations, each active node has at most $O(\log n)$ active neighbors that are not spoiled, and each active node has at most $O(\log^2 n)$ active and spoiled neighbors, with probability $1 - n^{-10}$.
\end{lemma}
\begin{proof}
	Let us call an iteration $i$ \emph{successful} if both \cref{lem:Bi} and \cref{lem:Ai} hold.
	This occurs with probability at least $1 - 2 \cdot n^{-11}$.
	By a union bound, the probability that all $T$ iterations are successful is at least $1 - n^{-10}$.
	Thus, with the same probability, we have that $A(T)$ and $B(T)$ hold.
	From $B(T)$ we get that every node has at most $\frac{\Delta}{2^{\log \Delta - \log \log n + 1}} = O(\log n)$ active and non-spoiled neighbors that remain, proving the first part of the statement.
	The second part follows directly from $A(\log \Delta)$, by observing that $T \leq \log n$.
\end{proof}
Thus, our algorithm achieves the desired guarantees for the maximum degree in the remaining graph.
What remains is to show that it can be implemented with low worst-case energy complexity.

\paragraph{Low Energy Complexity}
We now describe how the algorithm can be implemented with just $O(\log \log n)$ worst-case energy complexity.
Observe, that since each node is marked at most once, there is also at most one round where a node can join the independent set.
Further, since all marking probabilities are fixed at the beginning of the algorithm, each node $v$ can find its round $r_v$ in which it is sampled before the algorithm even starts, or observe that it is never sampled.
So we have a setting where each node has (at most) one round $r_v$ assigned to it.
Ideally, we would want $v$ to only be awake in round $r_v$, however, it needs some additional information:
$v$ needs to know if any of its neighbors joined the independent set in a previous round, meaning that $v$ is inactive in $r_v$.

We can do so by first adding a (third) sub-round to each round, which will be used to exchange this information.
Thus, we want for each node a set of rounds $S_v$ in which it should be awake, such that for any neighbor $u$ with $r_u \leq r_v$, there is a round $k \in S_v \cap S_u$ with $r_u \leq k \leq r_v$, in which both $v$ and $u$ are active, and $v$ can thus learn if $u$ joined the independent set in this new sub-round.
In the first two sub-rounds of round $k$, only nodes $w$ with $r_w = k$ are awake; all other nodes are only awake for the last sub-round.
Further, we want each set $S_v$ to be as small as possible, since the maximum size of them will correspond to the worst-case energy complexity of our algorithm.
And finally, instead of assigning a set $S_v$ to each node $v$, we can find a set $S_k$ for each round $k$, and node $v$ will be assigned the set $S_k$ with $r_v = k$.

The problem is now formalized in the following lemma.
Similar results have been proven in \cite{barenboim2021deterministic, augustine2022brief}, and such schedules have also been referred to as a ``virtual binary tree''.
For this work to be self-contained, we reprove the statement differently, with a concise recursion.
\begin{lemma}
	\label[lemma]{lem:devAndConq}
	Let $T$ be the total number of rounds, numbered $1, \dots, T$.
	For each round $k$, we can construct a set of rounds $S_k$ with $\Abs{S_k} = O(\log T)$, such that for any two rounds $i \leq j$ there is a round $l$ with $i \leq l \leq j$ and $l \in S_i \cap S_j$.
\end{lemma}
\begin{proof}
	We will use a divide-and-conquer approach to construct the sets $S_k$.
	To generalize the description, say our rounds are numbered $L, \dots, H$.
	Initially, we have $L = 1$ and $H = T$.
	Let $M = L + \lfloor (H-L)/2 \rfloor$.
	In the first step, our goal is to allow communication between all rounds $i \leq M$ and $j \geq M + 1$.
	We can do so by having everyone be awake in round $M$, i.e., adding $M$ to each set $S_k$.
	This takes care of all such pairs $i$ and $j$.
	Thus, all pairs of rounds $i \leq j$ that do not satisfy the lemma statement yet must have either $i \leq j \leq M$ or $M + 1 \leq i \leq j$.
	So we can do the same procedure on each of those two remaining parts.
	For the first part we set $L_1 = L$ and $H_1 = M$, and for the second part we set $L_2 = M+1$ and $H_2 = H$.
	This process terminates after $\lceil \log n \rceil$ levels of recursion, and in each level, we add one element to the set $S_k$ for each round $k$.
	Thus, we have $\Abs{S_k} \leq \lceil \log n \rceil$.
\end{proof}

\paragraph{Wrap-Up}
Now we finally have all the necessary parts to prove the main lemma of this subsection.
\begin{proof}[Proof of \Cref{lem:alg1p1}]
	First, we will review how the algorithm is executed in a distributed fashion.
	Initially, each node $v$ performs the sampling and finds a round $r_v$ in which it was sampled, or it does not sample itself in any round.
	In the latter case, the node will stay asleep for the entire algorithm.

	All nodes $v$ that have found a round $r_v$ will use \Cref{lem:devAndConq} with $T = c \log n \cdot (\log \Delta - \log \log n)$ to compute a set $S_{r_v}$ for all rounds they will be awake.
	Let $i \in S_{r_v}$ be a round in which $v$ is awake.
	If $i < r_v$, then $v$ will only be awake in the third sub-round and listen if any of its neighbors already joined the MIS in a previous round.
	For the case $i = r_v$, it will be awake in all three sub-rounds.
	In the first two, it will execute the algorithm together with all other neighbors that were also sampled in this round, and then in the third round, it will inform its neighbors if $v$ was added to the independent set.
	Finally, if $i > r_v$, then $v$ will only be awake in the third sub-round again and inform its neighbors if it joined the independent set in round $r_v$.

	It remains to argue that this algorithm achieves all the claimed guarantees.
	From \Cref{lem:p1degree}, we know that each node has at most $O(\log n) + O(\log^2 n) = O(\log^2 n)$ active neighbors that remain.
	We have $T = c \log n \cdot (\log \Delta - \log \log n) = O(\log^2 n)$ total rounds, and from \Cref{lem:devAndConq} we know that the worst-case energy complexity is $O(\log T) = O(\log \log n)$.
	Finally, the algorithm works in the \congest model since all messages that are exchanged are only about whether or not nodes were marked or included in the independent set, which can be sent using just single-bit messages.
\end{proof}

\subsection{Phase II}
\label{sec:alg1p2}
From the previous phase, we know that we now have the following setting:
A set of nodes that is inactive, meaning either the node or one of its neighbors is in the independent set, and a set of remaining nodes that are still active, which induce a graph with maximum degree $\Delta_2 = O(\log^2 n)$.
The inactive nodes will be ignored from now on (since they already satisfy the condition for an MIS).
We focus on the graph induced by the active nodes.

The goal of this phase is to remove an additional set of nodes that form an independent set, whose removal, together with their neighbors, leaves only connected components of size $\poly (\log n)$.
Further, each of these components only has diameter $O(\log n)$.
This has already been obtained in previous work by Ghaffari~\cite{gmis, ghaffari2019missmall}.
We restate the results we will use.
Note that these results assume that nodes are awake in all rounds.
This will not negatively impact our running time since for graphs of maximum degree $\Delta_2 = O(\log^2 n)$, the time complexity is only $O(\log \log n)$, and we can indeed keep all nodes awake during this part of the algorithm.
\begin{lemma} [Ghaffari~\cite{gmis, ghaffari2019missmall}]
	\label[lemma]{lem:alg1p2}
	There is an algorithm that, given a graph of maximum degree $\Delta = \poly (\log n)$, computes an independent set $S$ of nodes, and a clustering of all nodes that are not in or adjacent to $S$ into clusters of diameter $O(\log \log n)$.
	Each connected component of the graph that remains after the removal of $S$ and all its neighbors contains at most $\poly (\log n)$ nodes and $O(\log n / \log \log n)$ clusters.
	Furthermore, the algorithm runs in $O(\log \Delta)$ rounds of the \congest model and succeeds with probability at least $1 - n^{-10}$.
\end{lemma}

\subsection{Phase III}
\label{sec:alg1p3}
Lastly, we have to deal with the remaining small components.
So we describe what we do for one component, but all components can be handled simultaneously.
Recall from the previous section that each node is part of a cluster of diameter at most $O(\log \log n)$ and each connected component contains at most $O(\log n / \log \log n)$ clusters.
Our goal will be to merge all these clusters into one cluster per component.
If our message size would be unlimited, this would already solve our problem since we could just exchange the topology of the entire component, and have each node run a sequential algorithm to obtain the result.
Since we want our algorithms to work in the \congest model, we will have to use an additional step, also used in \cite{ghaffari2019missmall}, to finally compute a maximal independent step.
That is, we run several randomized algorithms, each with an overall small success probability in parallel, and use the fact that clusters have low diameter to find a successful execution, which exists with high probability.
In particular, we will show the following:
\begin{lemma}
	\label[lemma]{lem:alg1p3}
	There is an algorithm that, given a graph of at most $\poly (\log n)$ nodes clustered into $O(\log n / \log \log n)$ clusters, each with diameter $O(\log \log n)$, finds an independent set, which is maximal with probability $1 - n^{-10}$.
	The algorithm runs in $O(\log n \cdot \log^2 \log n)$ rounds of the \congest model, and each node is awake for at most $O(\log \log n)$ rounds.
\end{lemma}
The main ingredient is an algorithm merging all clusters into one cluster per component, which is also associated with a spanning tree of diameter $O(\log n)$.
This will allow us to perform broadcast and convergecast on this spanning tree with $O(1)$ awake time per node in $O(\log n)$ rounds.
The problem of merging clusters in an energy-efficient manner has been studied before, e.g., for minimum spanning tree algorithms.
Prior work~\cite{barenboim2021deterministic, augustine2022brief, dufoulon2022sleeping} obtained various trade-offs between energy and time complexity for both deterministic and randomized algorithms.
All these algorithms follow the outline of Bor\r{u}vka's algorithm \cite{boruuvka1926jistem} for the minimum spanning tree problem, which was first used in the distributed setting by Gallager, Humblet, and Spira~\cite{gallager1983distributed}.

Compared to these previous works, we are interested in a different trade-off between time and energy complexity, and simply changing their parameters is not sufficient for us.
Further, we remark that while there are randomized algorithms achieving our desired (and better) trade-offs, these would not work in our setting.
Since the graphs only have size $\poly (\log n)$, the success probability of standard randomized algorithms is only $1 - 1/ \poly (\log n)$.
Thus, our algorithm will need to be deterministic.
We will first state our desired guarantees, then give a brief overview of how to communicate within clusters in an energy-efficient manner, and finally, present the full algorithm for connecting all clusters in a component.
\begin{lemma}
	\label[lemma]{lem:mergeClusters}
	There is an algorithm that, given a graph of at most $\poly (\log n)$ nodes clustered into $O(\log n / \log \log n)$ clusters, each with diameter $O(\log \log n)$ finds a spanning tree with diameter $O(\log n)$ of the graph, in $O(\log n \cdot \log^2 \log n)$ rounds of the \congest model, of which each node is awake for at most $O(\log \log n)$ rounds, and succeeds with probability at least $1 - n^{-10}$.
\end{lemma}

\paragraph{Energy Efficient Broadcast and Convergeast}
Before giving the algorithm proving \Cref{lem:mergeClusters}, we outline how, given a cluster associated with a rooted spanning tree of diameter $D$ and each node knowing its distance to the root and $D$, we can perform broadcast and convergecast in this cluster with $O(1)$ energy complexity and $O(D)$ time complexity.
We will describe the broadcast procedure; convergecast works similarly.
For a detailed and formal description of these procedures, we refer to \cite{augustine2022brief} or \cite{barenboim2021deterministic}, where the assumed structure of our clusters is called ``Labeled Distance Tree (LDT)'' and ``Distributed Layered Tree (DLT)'' respectively.

Note that in the standard broadcast in a tree, each node only performs an operation in two rounds.
In the first round, it receives the message from its parent, and in the second round, it sends it to its children.
Since we assume that a node $v$ knows its distance to the root $d_v$, these rounds are precisely $d_v$ and $d_v + 1$.
Thus, it suffices for $v$ to be awake in only these two rounds, which it can calculate beforehand, and be sleeping in the remaining rounds.
This means that each node is awake in $O(1)$ rounds, and as for the standard broadcast, the procedure requires $O(\log n)$ rounds.

With these procedures as subroutines, we can now give the algorithm proving \Cref{lem:mergeClusters}.
\begin{proof}[Proof of \Cref{lem:mergeClusters}]
	As a first step, we will ensure that each cluster has the required structure.
	For this, all nodes are awake in each round.
	Each cluster will find a root by electing a leader (i.e., the node of minimum identifier), after which the root will start a broadcast, building a spanning tree.
	Each node will learn its distance from the root.
	For the diameter of a cluster, we will always use $O(\log n)$ as an upper bound.
	Since we have $O(\log n / \log \log n)$ clusters of diameter $O(\log \log n)$, and we only ever merge clusters, this is indeed an upper bound of the diameter of any cluster.

	Then, we perform $O(\log \log n)$ iterations, in each of which we will reduce the number of remaining clusters by a constant factor.
	Some steps will be easier to describe in the following view:
	Let $H$ be the graph obtained by contracting each cluster into one node, and connecting two clusters in $H$ if they contain adjacent nodes.
	We will call this $H$ the \emph{cluster graph}.
	With this, we are ready to describe the steps in each iteration.

	\paragraph{Finding Outgoing Edges}
	As in Bor\r{u}vka's algorithm, each cluster will first select one incident edge as outgoing.
	This will be the edge to the neighboring cluster whose cluster identifier (which is just the identifier of its root) is minimal.
	In the case there are multiple edges to the same cluster, we choose one in an arbitrary but consistent way (e.g., by choosing the edge with the minimum identifier, where the edge ID is the concatenation of the ID of its incident nodes).
	To obtain an energy-efficient algorithm, this can be implemented as follows:
	In the first round, all nodes are awake and exchange the identifier of their cluster with all their neighbors.
	Then, each node chooses the cluster with the minimum identifier, and we start a convergecast in each cluster, always forwarding the neighboring cluster with the minimum identifier.
	Once this has reached the root, it will broadcast to all nodes whose edge was chosen.
	Then, we remove all edges that were not selected as outgoing from $H$.
	If two neighboring clusters chose the same edge as their outgoing edge, we will set this edge aside for now, and ignore the edges in both directions in $H$.
	Let the set of edges be $M$.
	By ignoring these parallel edges, and noting that outgoing edges always point to a cluster with the lower identifier, the remaining graph $H$ cannot contain any cycles and must be a forest of oriented trees.
	This observation concludes this step and we move on to the next one.

	\paragraph{High- and Low-Indegree Clusters}
	We call a cluster high indegree if 10 or more other clusters chose it as their outgoing neighbor.
	Otherwise, we say it has low degree.
	In the same way as before, a cluster can learn its indegree by each node being awake in the first round of this step, and if it is incident to the edge selected by its cluster, it will inform its neighbor about this selection.
	Then, we can perform a convergecast in each cluster, summing up the number of incoming edges.
	Finally, each cluster performs a broadcast such that all nodes in the cluster learn the indegree.
	If a cluster has high indegree, it will remove its outgoing edge from $H$ and \emph{accept} all remaining incoming edges.
	We call this set of accepted and non-removed edges $E_H$.
	This will be one of the sets of edges we use for merging clusters.
	This takes care of all edges incident to high indegree clusters, so from now on we will focus on edges between low-degree clusters.
	Let $H_L$ be the subgraph of $H$ induced by low indegree clusters.
	We are interested in only one property of $H_L$, namely that if we drop the orientations, $H_L$ has degree at most $10$ (at most 9 incoming and 1 outgoing edge).

	\paragraph{Maximal Matching on $H_L$}
	To compute a maximal matching on $H_L$, we will first compute a coloring of $H_L$.
	Before we describe the algorithm, we argue that we can simulate the execution of any \congest model algorithm on $H_L$ in the base graph using only one broadcast, one convergecast, and one additional round:
	First, each root broadcasts the messages it intends to send its neighbor to all nodes in its cluster.
	Then, each node that has an adjacent edge in $H_L$ sends the received message to the corresponding neighboring cluster.
	Finally, by a convergecast in each cluster, the root learns all messages.
	Since $H_L$ has maximum degree 10, there are at most 10 messages sent and received in each cluster, and this still works in the \congest model with $O(\log n)$-bit messages.

	For coloring $H_L$ we can now use Linial's algorithm \cite[Theorem 5.1]{linial92}.
	Given a graph of maximum degree $\Delta$ and a $k$ coloring, this algorithm returns a $O(\Delta^2 \log k)$ coloring in one \congest round.
	Performing two rounds of this will give us a coloring with at most $O(\log \log n)$ colors, where we used the cluster identifier as the initial coloring.
	Then, we go through the color classes one by one and use the given orientation of $H_L$.
	For each color class $i$, all clusters $C$ of color $i$ check whether or not they were matched to a cluster with lower color.
	If not, $C$ matches itself to one of its unmatched incoming neighbors, if such a neighbor exists.
	Let this computed set of edges be $M_L$.

	We argue that $M_L$ is a maximal matching in $H_L$.
	First, it is a matching because no cluster is adjacent to more than one edge.
	In each color class, no two clusters can be matched to the same cluster since each cluster that is not of the current color can only be matched to its out-neighbor, of which there is at most one.
	And since it is a proper coloring, there are no two neighboring clusters of the same color.
	Further, the matching is maximal since if a cluster $C$ is unmatched all its outgoing neighbors must be matched, as otherwise $C$ would have been matched to one of them.
	The same is true for the incoming neighbor since if it was not matched, it would have been matched to $C$.

	\paragraph{Merging Clusters}
	Now we can discuss how and which clusters we merge.
	We will merge along all edges from the sets $M$, $E_H$, and $M_L$ one by one, and finally, along a set of edges $R$ which we describe now:
	Each low-degree cluster $C$, that has no incident edge in $M \cup E_H \cup M_L$, chooses an arbitrary outgoing neighbor that has an incident edge in $M \cup E_H \cup M_L$ and adds this edge to $R$.
	Such a neighbor must exist: $C$ can only have low-degree neighbors by definition, and if one of them was not matched in $M_L$, it would contradict $M_L$'s maximality.
	This set $R$ induces a subgraph of $H$, where each connected component is a star since we connect two disjoint sets of clusters.

	Merging along all those sets, we get that each cluster is merged with at least one other cluster, meaning that the number of clusters decreases by a factor of two in each iteration.
	However, we still need to describe how we perform the merges:
	We show how we perform star-shaped merges, where we have a number of leaf clusters that are merging oriented towards a center cluster, which is the case for $E_H$ and $R$.
	Since both $M$ and $M_L$ are matchings, this is just a special case of star-shaped merges.
	We can merge clusters by just adding the edges connecting them to the spanning tree.
	Then for all edges $(u, v)$, such that $u$ is in the center cluster and $v$ is in a leaf cluster, we set $u$ as the parent of $v$ and set the distance of $v$ to the root $d_v$ to be $d_u + 1$.
	Then, we perform a convergecast in the leaf cluster, allowing us to set the distance to $v$ of all nodes that are between $v$ and the root of the leaf cluster.
    For computing the distance of all the remaining nodes to $v$, we perform one more broadcast from the original root of the leaf cluster, allowing $v$ to become the new root.

	\paragraph{Time and Energy Complexity}
	First, we analyze the time complexity.
	Since we half the number of clusters in each iteration, we have $O(\log \log n)$ iterations.
	In each step, except for computing the matching on $H_L$, we use a constant number of rounds for communication between clusters, plus a constant number of broadcasts and convergecasts.
	For computing the matching on $H_L$, we perform one broadcast and one convergecast, plus one additional round per color class.
	Since the cluster diameter is at most $O(\log n)$, the overall time complexity is $O(\log n \cdot \log^2 \log n)$.

	For the energy complexity, using the same argument as for the time complexity gives an upper bound of $O(\log^2 \log n)$.
	However, despite going through the $O(\log \log n)$ color classes, we still only need $O(1)$ awake rounds.
	Each cluster only has a constant number of neighboring clusters, and only needs to be awake in the rounds where these clusters could be matched to them, to learn whether or not it has been matched in future rounds.
	Thus, iterating through all the $O(\log \log n)$ color classes still only requires $O(1)$ awake rounds per node, and thus also $O(1)$ awake rounds per iteration suffice.
	Overall, we have $O(\log \log n)$ rounds where a node needs to be awake.
\end{proof}

\noindent\textbf{Computing the MIS}
From \Cref{lem:mergeClusters} we know that we can obtain a spanning tree of depth $O(\log n)$ in each cluster, and we know that we can perform broadcast and convergecast in each cluster with $O(1)$ energy complexity.
If we had messages of unlimited size, we could just perform one convergecast for the root to learn the topology of the entire cluster, compute an MIS locally, and then distribute this information to all nodes.
For the \congest model, this does not work, but we can use the same idea as in \cite{ghaffari2019missmall} to get around this issue, which we describe now.
\begin{proof}[Proof of \Cref{lem:alg1p3}]
	We first use the algorithm from \Cref{lem:mergeClusters} clusters to get a spanning tree of diameter $O(\log n)$ of the graph.
	To compute an MIS we will use Ghaffari's algorithm~\cite[Theorem 3.1]{gmis} with all nodes being awake for the entire algorithm.
	This algorithm can be executed using 1-bit messages and gives the following guarantee, restated for graphs with maximum degree $\poly (\log n)$.
	After $O(\log \log n)$ rounds of the algorithm, for a given node $v$, the probability that neither $v$ nor any of its neighbors have joined the MIS is at most $1/\poly (\log n)$.
	Since each component has $\poly (\log n)$ nodes, by a union bound the probability that any of them remains (i.e. has not joined the MIS or one of its neighbors joined the MIS) is at most $1 / \poly (\log n)$ still.
	With one execution of the algorithm using only 1-bit messages, we can run $O(\log n)$ independent executions of the algorithm in parallel, using $O(\log n)$ bit messages.
	Since the success of each of those executions is independent, the probability that none of them computes an MIS is less than $n^{-10}$.
	Each node $v$ can now check if an execution was successful locally.
	It checks if any neighbor $u$ of $v$ joined the MIS.
	Then, $v$ is successful if either $v$ joined the MIS and there is no such neighbor $u$, or $v$ did not join the MIS and there is a neighbor $u$ that did.
	Again, this can be done for each execution of the algorithm with 1-bit messages, thus for all executions in parallel in the \congest model.
	To check if one execution was successful for all nodes a convergecast with 1-bit messages suffices, where the aggregate function is just the logical AND.
	This means we can check all executions in parallel and the root of the spanning tree can pick the first successful one, and inform all nodes to use their output in this execution.
	Since at least one execution was successful with probability $1 - n^{-10}$, the algorithm succeeds with the same probability.
\end{proof}

\subsection{Putting Everything Together}
Having described all three phases, we can now prove our first main result.
\begin{proof}[Proof of \Cref{thm:alg1}]
	First, we execute the algorithm described in \Cref{lem:alg1p1}, and remove the independent set together with its neighbors from the graph.
	Since the degree of the remaining graph has now dropped to $O(\log^2 n)$, we can use the algorithm from \Cref{lem:alg1p2} now.
	This again computes an independent set, which we remove from the graph together with their neighbors.
	The remaining nodes are clustered into $O(\log \log n)$-diameter clusters, and each connected component of the remaining graph contains at most $\poly (\log n)$ nodes, and $O(\log n / \log \log n)$.
	By treating each remaining connected component as an independent graph, we can execute the algorithm from \Cref{lem:alg1p3} on each of them.
	The resulting set is maximal in the original graph, if the first two phases succeed, and the third phase computes an MIS in each cluster.
	By a union bound over all clusters, the latter occurs with probability at least $1 - n^{-9}$.
	Combining this with the success probabilities of the first two phases, the probability that the algorithm computes an MIS is $1 - n^{-8}$.

	The algorithm runs for $O(\log^2 n) + O(\log \log n) + O(\log n \cdot \log^2 \log n) = O(\log^2 n)$ rounds, and each node is awake for $3 \cdot O(\log \log n) = O(\log \log n)$ rounds.
\end{proof}

\section{Algorithm 2}
\label{sec:alg2}
In this section, we will describe an algorithm that has energy complexity $O(\log^2 \log n)$ and a time complexity of $O(\log n \cdot \log \log n \cdot \log^* n)$.
This will prove \Cref{thm:alg2}.
Compared to the previous the main change will be in phase I, which is an entirely different algorithm.
Phase II will be identical, and for Phase III we have the same algorithm, however, we will set the parameters differently to obtain a better time complexity, at the cost of a slightly worse energy complexity.

\subsection{Phase I}
\label{sec:alg2p1}

\newcommand{\deltaLB}{\Omega(\log^{20} n)}

As in the first algorithm, the goal of this phase is to find an independent set, such that removing it and its neighbors from the graph, reduces the maximum degree to $\poly (\log n)$.
The algorithm is made up of different iterations, in each of which we will go from a graph with maximum degree $\Delta$ to a graph with maximum degree $\Delta^{0.7}$, as long as $\Delta = \deltaLB$.
Every such iteration will require $O(\log n)$ rounds and $O(\log \log n)$ energy complexity.

The general idea is similar to a different version of Luby's algorithm:
In each round, a node $v$ which still has $\deg(v)$ \emph{active} neighbors is marked with probability $\frac{1}{2 \deg(v)}$.
For any edge that has both endpoints marked, we remove the marking of the endpoint with the lower degree, breaking ties arbitrarily.
Now each node $v$ that is still marked has no marked neighbor, so we add $v$ to the MIS and remove $v$ and all its neighbors from the graph.
It can be shown that this algorithm computes a maximal independent set in $O(\log n)$ rounds with high probability.

For our goal of achieving low energy complexity, we have an additional obstacle compared to the previous algorithm:
The marking probabilities depend on the current state of the graph, which in turn depends on the execution of the algorithm and the random choices in the previous steps.
Some of the issues of the previous section remain, namely that nodes can be marked $\Theta(\log n)$ times\footnote{Indeed,  e.g. on constant degree graphs this will occur with high probability for at least one node.}.

As before, in order to deal with these issues we also need to weaken the guarantees of the algorithm.
While Luby's algorithm computes an MIS, our algorithm will only be able to find an independent set, such that the remaining graph has maximum degree at most $\Delta^{0.7}$, as long as $\Delta = \deltaLB$.
This allows us to almost ignore all nodes of degree less than $\Delta^{0.6}$.
For the remaining high-degree nodes $v$, we can get an accurate estimate $\widetilde{\deg}(v)$ of the degree $\deg(v)$ by sampling their neighbors with probability $\frac{1}{\Delta^{0.5}}$.
Further, we also cap the marking probability at roughly $\frac{1}{\Delta^{0.6}}$.
This means we can pre-mark nodes with the same probability, and then later adjust it, while ensuring that only a small number of nodes are pre-marked.
By being able to do this at the beginning of the algorithm, the nodes can also set their awake schedule at this time.
To sum up, we now have two ways in which nodes can be sampled, and in both of them, the sampling probability is at most $\frac{1}{\Delta^{0.5}}$.
This allows us to use the same ideas as in the first algorithm, where we only consider a node sampled once and then discard it from the algorithm, regardless of if it was removed from the graph.
This will allow us to get the desired bounds on the energy complexity.

To sum up, we will show that there is an algorithm with the following guarantees:
\begin{lemma}
	\label[lemma]{lem:alg2p1}
	There is an algorithm that, given a graph of maximum degree $\Delta \geq \deltaLB$ computes an independent set $S$ of nodes, such that, with probability $1 - n^{-6}$, the removal of all nodes from $S$ together with their neighbors results in a graph of degree at most $\Delta^{0.7}$.
	The algorithm runs in $O(\log n)$ rounds and each node is awake for at most $O(\log \log n)$ rounds.
\end{lemma}
By repeatedly executing this algorithm for $O(\log \log \Delta)$ iterations, each time with a smaller maximum degree $\Delta$, until $\Delta \leq O(\log^{20} n)$, we can conclude:
\begin{corollary}
	\label[corollary]{cor:alg2p1}
	There is an algorithm that, given a graph of maximum degree $\Delta$ computes an independent set $S$ of nodes, such that, with probability $1 - n^{-5}$, the removal of all nodes from $S$ together with their neighbors results in a graph of degree at most $O(\log^{20} n)$.
	The algorithm runs in $O(\log n \cdot \log \log \Delta)$ rounds and each node is awake for at most $O(\log \log n \cdot \log \log \Delta)$ rounds.
\end{corollary}

\paragraph{Algorithm}
We describe one iteration of the algorithm.
Given a graph of degree $\Delta = \deltaLB$, it returns an independent set such that, w.h.p., the maximum degree in the remaining graph is $\Delta^{0.7}$.
This is then applied recursively until the degree drops below $\deltaLB$ and we move on to the second phase.
For simplicity, we first describe the algorithm as if all nodes are active in all rounds.
We discuss how it can be implemented with low energy complexity later in this section.

At the beginning of each iteration, all nodes are \emph{active}.
As a first step, they check if any of their neighbors joined the MIS in the previous iteration.
If they have such a neighbor, they are removed from the graph and become \emph{inactive}.
Then, they perform two types of sampling:
(A) for every round, a node $v$ flips a coin which comes up head with probability $\frac{1}{\Delta^{0.5}}$.
Once a node $v$ flips heads in round $i$, it is \emph{tagged} in round $i$ and will not participate in the rest of the sampling of type (A).
These nodes will be used to estimate the degree for setting the sampling probability.
For the sampling of type (B) each node $v$ performs the same process again, but the probability of flipping heads is $\frac{1}{2\Delta^{0.6}}$.
In the first round $i$ that $v$ flips heads, it is \emph{pre-marked} in round $i$ and again, will not participate in the sampling of the following rounds.
We use this sampling to compute MIS.
We call a node $v$ \emph{spoiled} in round $j$ if there is a round $i < j$, such that $v$ was sampled of either type (A) or (B) in round $i$.

The algorithm performs $O(\log n)$ of the following rounds, in each of which we remove some nodes, and then proceed on the remaining subgraph.
We split each round $i$ into three sub-rounds, each corresponding to one round of communication in the base graph.
In the first sub-round, all nodes that were tagged (i.e. sampled of type (A)), inform all their neighbors of this tagging.
Then, each node $v$ that was pre-marked keeps track of the number of neighbors that were sampled.
Let $A_v$ be this number of tagged neighbors.
The degree can now be estimated as $\widetilde{\deg}(v) = \Delta^{0.5} \cdot A_v$.
Since the node was pre-marked with probability $\frac{1}{2\Delta^{0.6}}$, to get a sampling with probability $\min \left\{ \frac{1}{2\Delta^{0.6}}, \frac{1}{5 \widetilde{\deg}(v)} \right\}$, we re-sample each node with probability $\min \left\{ 1, \frac{2\Delta^{0.6}}{5\widetilde{\deg}(v)}\right\}$ and call a node \emph{marked} if it was re-sampled.

In the second sub-round, each marked node $v$ informs its neighbors that it was marked, and also sends them its estimated degree $\widetilde{\deg}(v)$.
If a marked node $v$ learns that one of its neighbors $u$ was marked, it removes its marking if $\widetilde{\deg}(v) \leq \widetilde{\deg}(u)$, and keeps its marking if no such neighbor exists.
Finally, each node that kept its marking joins the independent set and informs all its neighbors about this decision in the final sub-round.
The nodes that joined the independent set, as well as their neighbors, are now considered \emph{inactive}.
Inactive nodes and nodes that are spoiled in iterations $j > i$, are removed from the graph, and we proceed to the next round.

Once an iteration has finished, we will show that the active (and possibly spoiled) nodes that have more than $4\Delta^{0.6}$ active and non-spoiled neighbors form an independent set, with high probability.
Thus, all active nodes will communicate with all their neighbors to get their exact number of active and non-spoiled neighbors.
If this number is more than $4 \Delta^{0.6}$, they will inform their neighbors of this fact.
Then, all nodes that have degree more than $4 \Delta^{0.6}$, and no neighbor with degree at least $4 \Delta^{0.6}$ will join the independent set.

\paragraph{Analysis}
Our analysis is close to a standard analysis by Alon, Babai, and Itai~\cite{alon1986fast}, as also outlined in~\cite{motwani1995randomized}.
However, due to our changes in the algorithm, each step in the analysis needs to be adapted.
Since we do not work with exact degrees, we first need to prove a standard result that for nodes of large degree, their degree estimate is close to their true degree.
Then, we will show that there are certain nodes, which we call \emph{good}, which are removed with constant probability in each step.
And finally, we can show at least half of the edges between high-degree nodes are good and are thus also removed with constant probability.
After $O(\log n)$ such iteration, this means that no edges between high-degree nodes remain, and the nodes that we add in the final step indeed form an independent set.

In order to try and avoid the ambiguity between the degree at the beginning and during the execution of the algorithm, we start by introducing the following notation.
\begin{definition}
	The \emph{remaining degree} $\deg_i(v)$ of a node $v$ in round $i$ is the number of active and non-spoiled neighbors, i.e. the number of neighbors that still participate in the algorithm.
\end{definition}
For the analysis, let us first prove that each node has an accurate estimation of its degree.
\begin{lemma}
	\label[lemma]{lem:DegreeEst}
	Let $\Delta \geq \deltaLB$.
	For any node $v$ that has remaining degree $\deg_i(v) \geq \Delta^{0.6}$ in a given round $i$, we have that the estimate $\widetilde{\deg}(v) \in [ 1/2 \deg_i(v), 2 \deg_i(v) ]$ with probability $1 - n^{-10}$.
	Further, for any node $u$ with remaining degree $\deg_i(u) < \Delta^{0.6}$ we have that $\widetilde{\deg}(u) \leq 2 \Delta^{0.6}$ with probability $1 - n^{-10}$ as well.
\end{lemma}
\begin{proof}
	Remember that we estimated the degree of $v$ as $\widetilde{\deg}(v) = \Delta^{0.5} \cdot A_v$, where $A_v$ is the number of tagged neighbors of $v$ in round $i$.
	We can write $A_v = \sum_{u \in N_i(v)} X_u$, where $X_u$ is the indicator variable for $u$ being tagged in round $i$, and we have $\Pr[X_u = 1] = \Delta^{-0.5}$ since $u$ is not spoiled.
	It holds that $\E[A_v] = \Delta^{-0.5} \cdot \deg_i(v)$, so the probability that $\widetilde{\deg}(v) \in [ 1/2 \deg_i(v), 2 \deg_i(v) ]$ is equal to the probability that $A_v \in [ 1/2 \E[A_v], 2 \E[A_v]]$.
	Since all $X_u$ are independent, we can use Chernoff's Bound to bound the inverse probability:
	\begin{align*}
		\Pr [ \Abs{A_v - \E[A_v]} \geq 1/2 \E[A_v] ] \leq 2 \exp(-\E[A_v]/12)
	\end{align*}
	From $\deg_i(v) \geq \Delta^{0.6}$ we have that $\E[A_v] \geq \Delta^{0.1} \geq \log^2 n$ since $\Delta \geq \deltaLB$.
	Thus, the probability is at most $n^{-10}$ and we have proven the first part of the statement.

	For the second part, we can follow the same line of argument, but we are interested in the probability of $A_u \leq 2 \Delta^{0.1}$ which is the same as the probability that $\widetilde{\deg}(u) \leq 2 \Delta^{0.6}$.
	We can again bound $\E[A_u]$ from $\deg_i(u) \leq \Delta^{0.6}$ by $\E[A_u] \leq \Delta^{0.1}$.
	We can use Chernoff's Bound again to bound the inverse probability:
	\begin{align*}
		\Pr \left[ A_u - \E[A_u] \geq \frac{\Delta^{0.1}}{\E[A_u]} \cdot \E[A_u]\right]
		 & \leq \exp \left(- \frac{ \frac{\Delta^{0.2}}{\E[A_u]^2} \E[A_u]}{2 + \frac{\Delta^{0.1}}{\E[A_u]}}  \right) \\
		 & = \exp \left( - \frac{\Delta^{0.2}}{2 \E[A_u] + \Delta^{0.1}} \right)                                       \\
		 & \leq \exp(\Delta^{0.1} / 3)
	\end{align*}
	Again, we have that $\Delta \geq \deltaLB$, so this probability is at most $n^{-10}$ proving the second part.
\end{proof}
This now allows us to define what kind of nodes we call \emph{good}, which have their name because they become inactive with constant probability in each iteration.
We note here that spoiled nodes can also be good, even if the naming could suggest otherwise.
This is necessary because once a node is marked it can become spoiled with constant probability.
So in the likely event that there is a high-degree node that becomes spoiled, it still needs to be removed from the graph, or at least its degree among active nodes needs to drop.
Thus, we still need to consider it in our analysis.
\begin{definition}
	We call a node \emph{good} in round $i$, if it is not inactive, has degree $\deg_i(v) \geq \Delta^{0.6}$ and more than a third of its neighbors $u$ have lower degree, i.e. have $\deg_i(u) < \deg_i(v)$.
	Otherwise, we call $v$ \emph{bad}.
	An edge $e$ is good if both its endpoints are good and bad otherwise.
\end{definition}
Unfortunately, we cannot make any statements about the number of good nodes compared to the number of bad nodes.
However, when looking at the edges, we can say that at least half the edges between high-degree nodes are good.
\begin{lemma}
	\label[lemma]{lem:goodEdges}
	At least half of all edges where both endpoints have degree at least $\Delta^{0.6}$ are good.
\end{lemma}
\begin{proof}
	First, for the analysis, orient all edges from the endpoint with the lower degree to the endpoint with the higher degree, breaking ties arbitrarily.
	Let us look at the subgraph induced by all nodes of degree at least $\Delta^{0.6}$.
	Note that all outgoing edges of nodes of degree at least $\Delta^{0.6}$ are preserved in this subgraph, as they point to nodes of even higher degree.

	Let us place 1 dollar on each bad edge incident to a node of degree at least $\Delta^{0.6}$, even if their other endpoint might have a lower degree.
	For each bad node, we have that at least 2/3 of its incident edges are outgoing and thus are in the relevant subgraph.
	So we can move the dollar on each bad edge, to its endpoint, which can then distribute it such that each of its outgoing edges receives 1/2 dollar.
	Thus, each edge in our subgraph receives 1/2 dollar, meaning at most half the edges can be bad.
	In other words, this means that at least half the edges are good, proving our statement.
\end{proof}
As the main technical part of this proof, we now show that good nodes are indeed removed with constant probability, as long as all degree estimates are accurate.
\begin{lemma}
	\label[lemma]{lem:goodNode}
	Assume that the statement of \Cref{lem:DegreeEst} holds (with probability 1) and that $\Delta \geq \deltaLB$.
	Then, for every good node $v$ with degree $\deg_i(v) \geq \Delta^{0.6}$, the probability that a neighbor of $v$ is marked and not unmarked is at least $(1 - e^{-1/30})/2$.
\end{lemma}
\begin{proof}
	First, we will show that the probability that no neighbor with a lower degree is marked is at most $e^{-1/30}$.
	Since $v$ is good, it has at least $1/3 \deg_i(v)$ neighbors with a lower degree.
	By \Cref{lem:DegreeEst}, each of those neighbors $u$ estimates their degree to be at most $\widetilde{\deg}(u) \leq 2 \deg_i(v)$.
	Combining this with $\deg_i(v) \geq \Delta^{0.6}$, we get that the marking probability of $u$ is
	\begin{align*}
		\min \left\{  \frac{1}{2\Delta^{0.6}}, \frac{1}{5 \widetilde{\deg}(u)} \right\} \geq \min \left\{  \frac{1}{2\Delta^{0.6}}, \frac{1}{10 \deg_i(v)} \right\} = \frac{1}{10 \deg_i(v)}.
	\end{align*}
	Thus, the probability that a lower degree neighbor is not marked is at most $1 - \frac{1}{10 \deg_i(v)}$.
	Since all neighbors are marked independently, the probability that none of them is marked is at most
	\begin{align*}
		\left( 1 - \frac{1}{10 \deg_i(v)}\right)^{1/3 \deg_i(v)} \leq \exp \left( - \frac{1/3 \deg_i(v)}{10 \deg_i(v)} \right) = e^{- 1/30}.
	\end{align*}

	Ideally, we could now just argue that a node $u$ which is marked is removed with constant probability, which is true, however combining these two statements results in issues with dependencies.
	Instead, we have to be a bit more careful about how the random choices are made.
	Let us look at the neighbors of $v$ one by one, and expose their random choices for the marking until we reach one node that is marked.
	By the previous paragraph, this process succeeds with probability $1 - e^{-30}$ before it runs out of neighbors.
	Now that we have such a marked neighbor $u$, we need to argue that it remains with probability at least 1/2.
	We have not yet revealed the random choices of any neighbors of $u$ that could be marked.
	If $u$ shares any neighbors with $v$, they were all not marked.

	There are at most $\deg_i(u)$ neighbors $w$ of $u$ remaining, which  could still be marked and have $\widetilde{\deg}(w) \geq \widetilde{\deg}(u)$.
	Let us first consider the case that $\deg_i(u) \geq \Delta^{0.6}$.
	In that case, by \Cref{lem:DegreeEst}, we have $\widetilde{\deg}(u) \geq 1/2 \deg(u)$.
	Thus, the probability that a given neighbor is marked is at most
	\begin{align*}
		\min \left\{ \frac{1}{2\Delta^{0.6}}, \frac{1}{5 \widetilde{\deg}(w)} \right\}
		\leq \min \left\{  \frac{1}{2\Delta^{0.6}}, \frac{1}{5 \widetilde{\deg}(u)} \right\}
		\leq \min \left\{  \frac{1}{2\Delta^{0.6}}, \frac{1}{2 \deg_i(u)} \right\}
		= \frac{1}{2 \deg_i(u)}.
	\end{align*}
	By a union bound over all neighbors of $u$ with a higher estimated degree, the probability that at least one of them is marked is at most $\deg_i(u) \cdot \frac{1}{2 \deg_i(u)} = 1/2$.
	This means that the probability that none of them is marked is at least $1/2$.

	For the second case that $\deg_i(u) \leq \Delta^{0.6}$ we have that each neighbor $w$ of $u$ is sampled with probability at most $\frac{1}{2 \Delta^{0.6}}$.
	Thus, the probability that any of them is sampled is at most $\Delta^{0.6} \cdot \frac{1}{2\Delta^{0.6}}$, which is at most 1/2.
	Finally, this means that in either case no higher degree neighbor of $u$ is marked with probability at least $1/2$ and thus $u$ is not unmarked.
\end{proof}
This directly means that every good edge is also removed with constant probability by at least one of its endpoints.
By linearity of expectation, this means that the expected number of edges drops by a constant factor.
However, this still assumes that the degree estimates are accurate, which might not be true.
But we have shown that this occurs with such a small probability that we can essentially ignore this case since its contribution to the expectation is negligible.
This intuition is formally proven in the following lemma.
\begin{lemma}
	\label[lemma]{lem:constantDecrease}
	For $\Delta \geq \deltaLB$, let $m_{i-1}$ be the number of edges whose endpoints both have degree at least $\Delta^{0.6}$ at the beginning of round $i$.
	Then, we have that at the end of round $i$, the expected number of such edges is $\E[m_i \mid m_{i-1}] \leq (1 - (1 - e^{-1/30})/8) \cdot m_{i-1}$.
\end{lemma}
\begin{proof}
	Let $\mathcal{E}$ be the event that the statement of \Cref{lem:DegreeEst} holds for all nodes.
	Since the statement of \Cref{lem:DegreeEst} holds with probability at least $1 - n^{-10}$ for each node $v$, by a union bound over all nodes, we have $\Pr[\mathcal{E}] \geq 1 - n^{-9}$ and $\Pr[\overline{\mathcal{E}}] \leq n^{-9}$.
	We can rewrite the expectation we are interested in as:
	\begin{align*}
		\E[m_i \mid m_{i-1}]
		= \E[m_i \mid m_{i-1} \text{ and } \mathcal{E}] \cdot \Pr[\mathcal{E}] + \E[m_i \mid m_{i-1} \text{ and } \overline{\mathcal{ E}}] \cdot \Pr[\overline{\mathcal{E}}]
	\end{align*}
	For the first term, we know from \Cref{lem:goodEdges} that at least $m_{i-1} / 2$ of all edges whose endpoints both have degree at least $\Delta^{0.6}$ are good.
	For those edges, we know that both their endpoints are good and have degree at least $\Delta^{0.6}$.
	Thus, by \Cref{lem:goodNode} each endpoint is removed with probability at least $(1 - e^{-1/30})/2$, and the same is true for the removal of each good edge.
	This means that for the expected value we have
	\begin{align*}
		\E[ m_i \mid m_{i-1} \text{ and } \mathcal{E} ] \leq m_{i-1} - m_{i-1}/2 \cdot (1 - e^{-1/30})/2 = (1 - (1 - e^{-1/30})/4) \cdot m_{i-1}.
	\end{align*}
	For $\Pr[\mathcal{E}]$, we will just need that it is at most 1.

	For the second term we trivially have that $\E[m_i \mid m_{i-1}] \leq m_{i-1}$, and $m_{i-1} \leq n^2$.
	Thus, we have $\E[m_i \mid m_{i-1} \text{ and } \overline{\mathcal{ E}}] \cdot \Pr[\overline{\mathcal{E}}] \ll (1 - e^{-1/30})/8 \cdot m_{i-1}$.

	Putting everything together, we have
	\begin{align*}
		\E[m_i \mid m_{i-1}]
		\leq (1 - (1 - e^{-1/30})/4) \cdot m_{i-1} + (1 - e^{-1/30})/8 \cdot m_{i-1} = (1 - (1 - e^{-1/30})/8) \cdot m_{i-1}.
	\end{align*}
\end{proof}
In words, \Cref{lem:constantDecrease} just says that in each round, the expected number of edges between nodes of degree at least $\Delta^{0.6}$ decreases by a constant factor.
Thus, for a large enough constant $C$, after $C \log n$ rounds, the expected number of such edges that remain is at most $n^{-10}$.
By Markov's inequality, the probability that this is at least one edge is at most $n^{-10}$.
Thus we can conclude the following:
\begin{corollary}
	\label[corollary]{cor:noHighDegree}
	For $\Delta \geq \deltaLB$, after $O(\log n)$ rounds, no edges between nodes of degree $\Delta^{0.6}$ remain with probability at least $1 - n^{-10}$.
\end{corollary}
Since we remove all nodes of degree at least $\Delta^{0.6}$ that remain after $O(\log n)$ rounds, this might seem like we have already proven \Cref{lem:alg2p1}.
However, we have discarded nodes that were spoiled, without them or any of their neighbors necessarily being in the computed independent set.
Thus, we still need to show that each node does not have too many neighbors that were spoiled.
\begin{lemma}
	\label[lemma]{lem:notManySpoiled}
	For $\Delta \geq \deltaLB$, each node $v$ has at most $4\Delta^{0.6}$ neighbors that are spoiled during the execution of the algorithm with probability at least $1 - n^{-10}$.
\end{lemma}
\begin{proof}
	We will show that each node $v$ has at most $4 \Delta^{0.6}$ neighbors that were ever sampled (in the form of being tagged or pre-marked).
	Since nodes can only be spoiled after they were sampled, this proves the lemma.

	Let us first look at the tagging of nodes.
	For each node $u$, let $X_u$ be the indicator variable that $u$ was ever tagged in one of the $O(\log n)$ rounds.
	Then, the number of tagged neighbors of $v$ is $X = \sum_{u \in N(v)} X_u$.
	In each round, we have a probability of $\Delta^{-0.5}$ to tag $u$.
	Thus, the probability $\Pr[X_u = 1]$ that $u$ is tagged in any round is at most $O(\log n) \cdot \Delta^{-0.5} \leq \Delta^{-0.4}$ since $\Delta \geq \deltaLB$.
	This means that $\E[X] \leq \Delta^{0.6}$.
	Since all $X_u$ are independent, we can use Chernoff's Bound to obtain
	\begin{align*}
		\Pr[X \geq 2 \Delta^{0.6}]
		\leq \Pr \left[ X - \E[X] \geq \frac{\Delta^{0.6}}{\E[X]} \E[X] \right]
		 & \leq \exp \left( - \frac{\frac{\Delta^{1.2}}{\E[X]^2}\E[X]}{2 + \frac{\Delta^{0.6}}{\E[X]}} \right) \\
		 & \leq \exp \left( - \frac{\Delta^{1.2}}{2\E[X] + \Delta^{0.6}} \right)                               \\
		 & \leq \exp \left( - \Delta^{0.6}/3 \right) \ll n^{-11},
	\end{align*}
	where we again used $\Delta \geq \deltaLB$ in the last inequality.
	To sum up, this means that at most $2 \Delta^{0.6}$ neighbors of $v$ are ever tagged during the algorithm with probability at least $1 - n^{-10}$.

	For the marking, we can use the same kind of argument.
	Let $Y_u$ be the indicator variable that $u$ is pre-marked in any round.
	Since the pre-marking probability is capped, we have that $\Pr[Y_u = 1] \leq \frac{1}{2 \Delta^{0.6}}$.
	We can again define $Y = \sum_{u \in N(v)} Y_u$ as the number of neighbors of $v$ that are ever pre-marked.
	Since $\Pr[Y_u = 1] \leq \Pr[X_u = 1]$, we also have that $\Pr[Y \geq 2 \Delta^{0.6}] \leq \Pr[X \geq 2 \Delta^{0.6}] \leq n^{-11}$.

	Finally, the probability that either the number of tagged nodes or the number of pre-marked nodes exceeds $2 \Delta^{0.6}$ is at most $2 n^{-11} \leq n^{-10}$.
	Thus, the probability that the number of spoiled neighbors is less than $4 \Delta^{0.6}$ is at most $1 - n^{-10}$.
\end{proof}
With this additional fact, we can finally prove \Cref{lem:alg2p1}:
\begin{proof}[Proof of \Cref{lem:alg2p1}]
	First, note that the algorithm always computes an independent set.
	Because we add an independent set of nodes in each iteration and then remove all their neighbors, the union of all those sets is still an independent set.

	Next, let us consider the maximum degree in the remaining graph.
	In the last step, all nodes of degree at least $4 \Delta^{0.6}$ are removed if they form an independent set, which by \Cref{cor:noHighDegree} happens with probability at least $1 - n^{-10}$.
	From \Cref{lem:notManySpoiled} we have that the number of spoiled neighbors (even inactive ones) is at most $4 \Delta^{0.6}$ as well, with probability at least $1 - n^{-10}$.
	Thus, with probability $1 - 2 n^{-10} \geq 1 - n^{-6}$ the total number of remaining neighbors is at most $8 \Delta^{0.6} \ll \Delta^{0.7}$.

	It remains to prove the time and energy complexities.
	For the time complexity, the algorithm runs for $O(\log n)$ rounds, each of which corresponds to three rounds of communication.
	This communication is possible in the \congest model since both the tagging and decision to join the independent set can be exchanged with 1-bit messages.
	For the second sub-round, it is sufficient for a node $v$ to send the number of tagged neighbors $A_v$ to each neighbor $u$.
	All other computations can then also be performed by $u$, and since $A_v$ is an integer and at most $n$, this can be sent in a $O(\log n)$-bit message.
	Thus, the algorithm runs in $O(\log n)$ rounds of the \congest model.

	For the energy complexity, we proceed in the same way as in phase I of the first algorithm.
	In the beginning, each node $v$ performs the sampling for being tagged or pre-marked.
	Let $r_v$ be the first round in which $v$ is sampled of either type.
	Note that $v$ can also be sampled of both kinds in the same round.
	Since this is the first round that $v$ is sampled, it cannot be spoiled yet, but it still needs to know whether or not it is still active.
	We can again do this by using \Cref{lem:devAndConq} with $T = O(\log n)$ to get a set $S_{r_v}$ of rounds in which a node $v$ will be awake.
	Then, we add a fourth sub-round to each round, in which nodes will inform their neighbors whether or not they joined the independent set in a previous round.
	This allows us to execute the same algorithm, but with each node being awake in only $O(\log \log n)$ rounds, which proves our bound on the energy complexity.
\end{proof}

\subsection{Phase III}
\label{sec:alg2p3}
While the second phase is identical to the first algorithm, there is a slight change in the third phase.
We are now interested in a different trade-off between round and energy complexity.
Since the first phase has energy complexity $O(\log^2 \log n)$, we want to now minimize the round complexity of the third phase, while still keeping the energy complexity below $O(\log^2 \log n)$.

The best such trade-off we are able to achieve is the following:
A round complexity of $O(\log n \cdot \log \log n \cdot \log^* n)$ and an energy complexity of $O(\log \log n \cdot \log^* n)$.
This result has already been proven in (the full version of) previous work \cite[Theorem 5.2]{barenboim2021deterministic, barenboim2021deterministicarxiv} so we refer to their paper for a full proof.
Since the adaptation to our algorithm is minor, we also give an outline here.

In the step, where we compute a maximal matching on $H_L$, we will compute a coloring with $O(1)$ colors instead of $O(\log \log n)$ colors, which will save an $O(\log \log n)$ factor in the running time.
This can be achieved by running Linial's algorithm for $O(\log^* n)$ rounds instead of just 2.
However, it comes at the cost that in each round of Linial's algorithm we need to communicate between neighboring clusters.
One such round of communication requires each node to be awake for $O(1)$ rounds, which are $O(\log^* n)$ rounds for the coloring step, giving the final complexities.

\subsection{Putting Everything Together}

With the described changes to the first and third phases of the algorithm, we can give the proof of our second main result.
\begin{proof}[Proof of \Cref{thm:alg2}]
  First, we execute the algorithm from \Cref{lem:alg2p1}, and remove the independent set together with its neighbors from the graph.
  Since the degree of the remaining graph has now dropped to $O(\log^{20} n)$, we can use the algorithm from \Cref{lem:alg1p2} now.
  This again computes an independent set, which we remove from the graph together with their neighbors.
  The remaining nodes are clustered into $O(\log \log n)$-diameter clusters, and each connected component of the remaining graph contains at most $\poly \log n$ nodes, and $O(\log n / \log \log n)$.
  We treat each remaining connected component as an independent graph, and compute an MIS on each of them in parallel.
  By replacing the algorithm from \Cref{lem:mergeClusters} with the algorithm from \cite[Theorem 5.2]{barenboim2021deterministic} in the proof of \Cref{lem:alg1p3}, we get the same guarantees as in \Cref{lem:alg1p3}, but in $O(\log n \cdot \log \log n  \cdot \log^* n)$ rounds, of which a node is awake for $O(\log \log n \cdot \log^*n)$ of them.
  This algorithm is then executed on all connected components in parallel.
  The resulting set is maximal in the original graph, if the first two phases succeed, and the third phase computes an MIS in each cluster.
  By a union bound over all clusters, the latter occurs with probability at least $1 - n^{-9}$.
  Combining this with the success probabilities of the first two phases, the probability that the algorithm computes an MIS is $1 - n^{-8}$.

  The algorithm runs for $O(\log n \cdot \log \log n) + O(\log \log n) + O(\log n \cdot \log \log n \cdot \log^* n) = O(\log n \cdot \log \log n \cdot \log^* n)$ rounds, and each node is awake for $O(\log^2 \log n)$ rounds.
\end{proof}

\section{Constant Average Energy Complexity}
\label{sec:average}
In this section, we show that the previous algorithms can be extended to have $O(1)$ average energy complexity, with high probability.
There will be two parts to achieve this:
First, we show that Phase I of both algorithms already has constant average energy complexity.
Second, we provide a new algorithm that works in between phases I and II.
Its goal is to reduce the number of nodes to $O(n / \log^2 \log n)$, which will mean that the algorithms from phases II and III only incur a total energy complexity of $O(n)$ on these remaining nodes, which is $O(1)$ averaged over all nodes.

For this, we will change the parameters of the algorithm from \Cref{sec:alg1p1}, with the goal of further reducing the degree in the remaining graph to $\poly \log \log n$.
However, this will come at the cost of some nodes failing and having to be discarded from the algorithm.
We will show that the probability that a single node fails is at most $1 / \poly \log n$, which (with some extra steps due to dependencies between nodes) will allow us to conclude that at most $O(n / \log n)$ nodes fail.
We will ignore these failed nodes, and deal with them later.
All nodes that do not fail will succeed and leave us with a graph of degree $\poly \log \log n$.

For the successful nodes, that now have degree at most $\poly \log \log n$, we use an algorithm from \cite{ghaffari2022average}, with the parameters adapted to our setting, to reduce the number of remaining nodes to $n / \log^2 \log n$.
Together with the failed nodes, this will mean that $O(n \log^2 \log n)$ nodes remain, in a subgraph of degree at most $\poly(\log n)$, so we can use the algorithms from phases II and III while still having an average energy complexity.

\subsection{Phase I}
\label{sec:constp1}
For the first algorithm, we have that maximum sampling probability per round is $\frac{\Delta}{2^T}$ for $T = \log \Delta - 2 \log \log n$, which is $\frac{1}{\log^2 n}$.
The probability that a node is ever sampled is thus at most
\begin{align*}
	\sum_{i = 0}^{T} c \cdot \log n \cdot \frac{\Delta}{2^i} \leq c \cdot \log n \cdot \frac{2}{\log^2 n} \leq 2c / \log n.
\end{align*}
Since nodes are sampled independently, by Chernoff's Bound, the number of nodes sampled is $O(n / \log n)$ with high probability.
And since each sampled node is awake for $O(\log \log n)$ rounds, the average awake time, or energy complexity, is $O(1)$.

For the second algorithm, we have that the recursion stops once we reach a graph with maximum degree $O(\log^{20} n)$.
Between tagging and pre-marking, the probability that a node is sampled of either kind in a given round is at most $\Delta^{-0.5} \leq 1 / \log^{10} n$.
Thus, the probability that a node is ever sampled is at most $\Omega( 1 / \log^9 n)$.
Again, by Chernoff's Bound, the number of nodes sampled is $O(n / \log n)$ with high probability.
And again, since each sampled node is awake for at most $O(\log \log n)$ rounds, the average energy complexity is $O(1)$.

\subsection{Phase I-II}
As with Phase II previously, this will be the same for both algorithms.
However, unlike before, this section will be the main technical and novel part to obtain constant average energy complexity.
We will give an algorithm that removes all but $O(n / \log^2 \log n)$ nodes from the graph.
Executing any algorithm with worst-case energy complexity $O(\log^2 \log n)$ on these remaining nodes, then only adds $O(1)$ to the total node-averaged energy complexity.
Thus, the following will directly allow us to get algorithms with node-averaged energy complexity $O(1)$.
\begin{lemma}
	\label[lemma]{lem:reduceNodes}
	There is an algorithm that, given a graph of maximum degree $\Delta_2 = \poly \log n$ computes an independent set $S$, such that removing all nodes in or adjacent to $S$ leaves at most $O(n / \log^2 \log n)$ nodes.
	The algorithm runs in $O(\log^2 \log n)$ clock rounds, has $O(\log \log n)$ worst-case, and $O(1)$ node-averaged energy complexity.
\end{lemma}
There are two main steps, the first has the goal of reducing the degree to $\poly \log \log n$, which we will be able to achieve for all but a small number of nodes.
With the maximum degree this low, we can use an algorithm from \cite{ghaffari2022average} to remove all but $O(n / \log^2 \log n)$ nodes.
Together with the nodes that failed in the first step, there will still be $O(n / \log^2 \log n)$ nodes that remain in total.

For the first step, we formally show an algorithm achieving the following.
\begin{lemma}
	\label[lemma]{lem:moreDegReduced}
	There is an algorithm that, given a graph of maximum degree $\Delta_2 = \poly \log n$ computes an independent set $S$ and a partition of all nodes that are not in or adjacent to $S$ in the following two sets:
	\begin{enumerate}
		\item[A:] The nodes in $A$ form an induced subgraph of degree at most $O(\log^{100} \log n)$, and
		\item[F:] Each node $v$ is part of the set of failed nodes $F$ with probability $1 / \log^{10} n$.
	\end{enumerate}
	The algorithm runs in $O(\log^2 \log n)$ rounds, each node is awake for at most $O(\log \log \log n)$ rounds, and the average energy complexity is $O(1)$ with high probability.
\end{lemma}
\begin{proof}
	The algorithm is almost the same as in the proof of \Cref{lem:alg1p1}, with three key differences:
	First, we only execute $O(\log \log n)$ rounds in each iteration, and, second, we stop once we reach iteration $T = \log \Delta_2 - 100 \log \log \log n$.
	The third, and biggest, difference is that at the end of each iteration, all nodes are awake for three rounds.
	In the first round, they learn if they are inactive, due to one of their neighbors joining the independent set.
	In the second round, they inform their neighbors of their current status, i.e. whether or not they are still active or spoiled.
	Then, each node can figure out if it is \emph{failed}.
	We now call a node $v$ \emph{failed} at the end of iteration $i$, if it has (A) more than $(i + 1) \cdot C \log \log n$ active and spoiled neighbors, or (B) more than $\frac{\Delta}{2^{(i+1)}}$ active and non-spoiled neighbors.
	Finally, nodes that failed add themselves to $F$ and inform their neighbors of this in the third round
	For all future rounds, we consider nodes that have failed as inactive.

	Note that the degree guarantee of the nodes in $A$ trivially holds.
	Since $i \leq \log \Delta_2 - 100 \log \log \log n$, all nodes that are not failed have degree at most $O(\log^{100} \log n)$
	So it remains to be shown that nodes fail with probability at most $1 / \log^{10} n$.
	We focus on one node $v$ and one iteration $i$, and show that if $v$ is not failed at the beginning of the iteration, it is failed with probability at most $1 / \log^{11} n$ at the end of $i$.
	There are two ways in which a node could fail, (A) and (B); we will first consider (B).
	In this case, $v$ has more than $\frac{\Delta}{2^{(i+1)}}$ active and non-spoiled neighbors remain in each round of iteration $i$.
	In each such round, as in \Cref{lem:Bi}, $v$ is removed with probability at least $1/40$.
	Over $O(\log \log n)$ rounds, the probability that this never occurs is at most $1 / \log^{12} n$ since different rounds are independent.

	The other reason a node can fail is (A) more than $(i + 1) \cdot C \log \log n$ of its neighbors are spoiled.
	Since $v$ was not spoiled at the end of the previous iteration, this can only occur if more than $C \log \log n$ neighbors are spoiled in this iteration.
	In the same way as in \Cref{lem:Ai}, we can also bound this probability using Chernoff's Bound to be at most $1 / \log^{12} n$.

	Thus, in one iteration, a node fails with probability at most $1 / 2 \log^{12} n \leq 1 / \log^{11} n$.
	Over all $O(\log \log n)$ iterations, this probability is at most $1 / \log^{10} n$ as desired.

	For the time complexity, we have $O(\log \Delta) = O(\log \log n)$ iterations of $O(\log \log n)$ rounds, which is at most $O(\log^2 \log n)$ rounds in total.
	For the worst-case energy complexity, we now have a simpler argument:
	Each node is only sampled in one iteration, thus it can just be awake for all those $O(\log \log n)$ rounds, and sleep for the other iterations.
	In addition to that, nodes are also awake at the end of each iteration for 3 rounds.
	Since there are $O(\log \log n)$ iterations, there are a total of $O(\log \log n)$ rounds that a node that was ever sampled is awake.

	So for the average energy complexity, we still need to argue that the number of nodes that are sampled is at most $O(n / \log \log n)$.
	As in the previous subsection, the probability that node is ever sampled is at most $O(\log \log n) \cdot \frac{2}{\log^{100} \log n} \leq 1 / \log \log n$.
	Since nodes are sampled independently, by Chernoff's Bound, the number of sampled nodes is at most $O(n / \log \log n)$ with high probability, concluding our proof.
\end{proof}
We now know that each node joined $F$ with probability at most $1 / \log^{10} n$, however, these events are not independent, so we cannot directly use the standard version of Chernoff's Bound.
What will still make it possible to show that there are at most $O(n / \log n)$ nodes in $F$ with high probability is that nodes at distance $\Omega(\log^2 \log n)$ are independently added to $F$.
We can think of the algorithm as each node first producing a random string that it will then use for all its random choices during the algorithm.
Then, the output of the algorithm at a node $v$ after $T$ rounds can also be described as a mapping of the $T$-hop neighborhood, and all random strings within this neighborhood, to the output of $v$.
Thus, all nodes at a distance more than $2T$ will have their outputs be independent of each other.

With this observation, we can now formally state and prove the result.
\begin{lemma}
  \label[lemma]{lem:notManyFailed}
	The number of failed nodes after the algorithm of \Cref{lem:moreDegReduced}, i.e. the size of the set $F$, is at most $O(n / \log n)$ with high probability.
\end{lemma}
\begin{proof}
	Let $X_v$ be the indicator variable for the event that $v$ failed.
	We know that $\E[X_v] \leq 1 / \log^{12} n$.
	Since the input graph has degree $\poly \log n$, and the algorithm runs for $O(\log^2 \log n)$ rounds, each $X_v$ is independent of all but at most $\log^{O(\log^2 \log n)} n$ random variables $X_u$ for other nodes $u$.
	This now allows us to use the following extension of Chernoff's Bound.
	\begin{theorem}[Chernoff Bound's Extension to Bounded-Dependency Cases~\cite{pemmaraju2001equitable}]
		\label[theorem]{thm:Chernoff-Extension}
		Let $X_1$, \dots, $X_N$ be binary random variables, and let $a_1, \dots, a_n$ be coefficients in $[0,1]$.
		Suppose that each $X_i$ is mutually independent of all other variables except for at most $\Delta'$ of them.
		Let $\bar{X}=\sum_{i=1}^{N} a_i \cdot X_i$, and $\mu=\mathbb{E}[\bar{X}].$
		Then, for any $\delta \geq 0$, we have $$\Pr\big(|\bar{X} - \mathbb{E}[\bar{X}]| \geq \delta \mathbb{E}[\bar{X}]\big) \leq 8(\Delta'+1) \cdot exp(-\frac{\mathbb{E}[\bar{X}] \delta}{3(\Delta'+1)}).$$
	\end{theorem}
	Applied to our setting, we have that $X_i = X_v$ all nodes $v$, and $a_i = 1$.
	Since we have $\E[\bar{X}] \leq n / \log^{10} n$, the probability that more than $O(n / \log n)$ nodes are in $F$ is at most $\Pr [ \Abs{X - \E[\bar{X}] } \geq C \cdot n / \log n ]$.
	Thus, for the extended Chernoff's Bound we have dependency degree $\Delta' = \log^{O(\log^2 \log n)} n$, and $\delta$ such that $\delta \E[\bar{X}] =  C \cdot n / \log )$.
	With these parameters, the probability bound given by \cref{thm:Chernoff-Extension} is $n^{-c}$ for any constant $c$, and a sufficiently large constant $C$.
	Thus, with high probability there are at most $O(n / \log n)$ failed nodes.
\end{proof}
So what remains is to show that in a graph of degree $O(\log^{100} \log n)$ we can compute an independent set, such that removing it together with its neighbors only leaves $O(n / \log^2 \log n)$ remaining nodes.
For this, we will use the algorithm from \cite[Section 3.2]{ghaffari2022average}.
While they do not state the guarantees of their algorithm in its full generality, they state their proof in a way such that they directly also prove the following:
\begin{lemma}[Lemma 3.2 in \cite{ghaffari2022average}]
  \label[lemma]{lem:reduceNodesOld}
  There is an algorithm that, given a graph of maximum degree $d = O(\log^2 n)$, and a parameter $k = O(\log \log n)$ for the number of iterations, finds an independent set $S$, such that removing $S$ and its neighbors from the graph leaves at most $O(n / 2^k)$ with high probability.
  The algorithm runs in $O(k \cdot \log^3 d)$ rounds of the \congest model, and has node-averaged energy complexity $O(1)$, with high probability.
\end{lemma}

\paragraph{Wrap-Up}
With these parts, we can now finally give the proof of the main result of this subsection.
\begin{proof}[Proof of \Cref{lem:reduceNodes}]
  First, we will execute the algorithm from \Cref{lem:moreDegReduced}.
  This will result in each node either being removed from the graph or being in the sets $A$ or $F$.
  We will ignore the nodes in $F$ (i.e. they will not be awake for the remainder of the algorithm), and only consider nodes from $A$ for now.
  Since the nodes from $A$ induce a subgraph of maximum degree $O(\log^{100} \log n)$, we can use \Cref{lem:reduceNodesOld} with $d = O(\log^{100} \log n)$ and $k = 2 \log \log \log n$.
  This gives an algorithm that runs in $O(\log^4 \log \log n)$ rounds, has node-averaged energy complexity $O(1)$, and leaves at most $O(n / \log^2 \log n)$ nodes remaining, with high probability.
  Thus, together with the nodes from $F$, at most $O(n / \log^2 \log n)$ nodes remain, with high probability.

  For the total round complexity, we spend $O(\log^2 \log n) + O(\log^4 \log \log n) = O(\log^2 \log n)$ clock rounds.
  The worst-case energy complexity is $O(\log \log n) + O(\log^4 \log \log n) = O(\log \log n)$, where for the second algorithm, we just use the number of clock rounds as an upper bound on the number of rounds a node is awake.
  Finally, on average, each node is awake for $O(1)$ rounds in each of the two algorithms, with high probability.
  Thus the overall node-averaged energy complexity is $O(1)$ with high probability as well.
\end{proof}

\subsection{Putting Everything Together}

For the claimed algorithms with constant node-averaged energy complexity, we first execute the algorithms from \Cref{sec:alg1p1} or \Cref{sec:alg2p1} respectively.
As argued in \Cref{sec:constp1}, the node averaged energy complexity of these algorithms is constant.
Then, we use the algorithm from \Cref{lem:reduceNodes} in between phases I and II, which again only requires nodes to be awake for a constant number of rounds on average.
Since the number of nodes is now only $O(n / \log^2 \log n)$, any algorithm with worst-case energy complexity $O(\log^2 \log n)$ will only require $O(1)$ average energy complexity.
Thus, we can proceed with Phases II and III of the previous algorithms, as described in \Cref{sec:alg1p2} and \Cref{sec:alg1p3,sec:alg2p3}.
Finally, since the algorithm from \Cref{lem:reduceNodes} has worst-case energy complexity $O(\log \log n)$ and requires at most $O(\log^2 \log n)$ many clock rounds, the previous bounds on both the number of awake and the total number of rounds still hold.

\section{Conclusion and Open Problems}

In this work, we have given round efficient algorithms with low worst-case energy complexity for the maximal independent set problem.
Allowing for a worst-case energy complexity of $O(\log^2 \log n)$, the round complexity nearly matches the $O(\log n)$ round complexity of Luby's algorithm~\cite{luby86, alon1986fast}.
For the best-known energy complexity of $O(\log \log n)$, we also made progress on improving the round complexity to $O(\log^2 n)$.
And finally, we showed that all this can also be achieved with constant node-averaged energy complexity.

There are two natural ways, in which this could be improved.
The first is improving the round complexity while maintaining the best-known worst-case energy complexity, i.e., getting an algorithm with round complexity $O(\log n)$ and worst-case energy complexity $O(\log \log n)$.
The second way in which the results could be improved is by lowering the worst-case energy complexity to $o(\log \log n)$.
We note that $O(\log \log n)$ is possibly a barrier for algorithms that require a sort of ``update'' between all pairs of rounds, i.e., for algorithms that need to solve the problem described in \Cref{lem:devAndConq}.
As any MIS algorithm needs $\widetilde{\Omega}{\sqrt{\log n}}$ rounds~\cite{kuhn16_jacm}, using this as a subroutine would always require an energy complexity of $\Omega(\log \log n)$.
Formalizing this, or proving any $\omega(1)$ lower bound on the worst-case energy complexity, would also improve our understanding of energy-efficient distributed algorithms.

\bibliography{ref}
\bibliographystyle{alpha}
\appendix

\end{document}